\documentclass[letterpaper,11pt]{article}

\usepackage{amsmath}
\usepackage{amsfonts}
\usepackage{amssymb}
\usepackage{amsthm}
\usepackage{url,ifthen}
\usepackage{srcltx}
\usepackage{multirow}
\usepackage{boxedminipage}
\usepackage[margin=1.125in]{geometry}
\usepackage{nicefrac}
\usepackage{xspace}
\usepackage{graphicx}
\usepackage{xcolor}
\definecolor{DarkGreen}{rgb}{0.1,0.5,0.1}
\definecolor{DarkRed}{rgb}{0.5,0.1,0.1}
\definecolor{DarkBlue}{rgb}{0.1,0.1,0.5}
\usepackage{tikz}
\usetikzlibrary{shapes,arrows}

\usepackage[small]{caption}

\usepackage[pdftex]{hyperref}
\hypersetup{
    unicode=false,          
    pdftoolbar=true,        
    pdfmenubar=true,        
    pdffitwindow=false,      
    pdfnewwindow=true,      
    colorlinks=true,       
    linkcolor=DarkBlue,          
    citecolor=DarkGreen,        
    filecolor=DarkGreen,      
    urlcolor=DarkBlue,          
    pdftitle={},
    pdfauthor={},    
}

\def\draft{1} 

\def\submit{0} 

\ifnum\draft=1 
    \def\ShowAuthNotes{1}
\else
    \def\ShowAuthNotes{0}
\fi

\ifnum\submit=1
\newcommand{\forsubmit}[1]{#1}
\newcommand{\forreals}[1]{}
\else
\newcommand{\forreals}[1]{#1}
\newcommand{\forsubmit}[1]{}
\fi

\ifnum\ShowAuthNotes=1
\newcommand{\authnote}[2]{{ \footnotesize \bf{\color{DarkRed}[#1's Note:
{\color{DarkBlue}#2}]}}}
\else
\newcommand{\authnote}[2]{}
\fi

\newtheorem{theorem}{Theorem}[section]
\newtheorem{remark}[theorem]{Remark}
\newtheorem{lemma}[theorem]{Lemma}
\newtheorem{corollary}[theorem]{Corollary}

\newtheorem{claim}[theorem]{Claim}
\newtheorem{fact}[theorem]{Fact}

\theoremstyle{definition}
\newtheorem{definition}[theorem]{Definition}

\newtheorem{condition}[theorem]{Condition}

\newcommand{\sectionlabel}[1]{\label{sec:#1}}

\usepackage[T1]{fontenc}
\usepackage{kpfonts}
\usepackage{microtype}

\newcommand{\Esymb}{\mathbb{E}}
\newcommand{\Psymb}{\mathbb{P}}

\DeclareMathOperator*{\E}{\Esymb}

\DeclareMathOperator*{\ProbOp}{\Psymb}
\renewcommand{\Pr}{\ProbOp}

\newcommand{\mper}{\,.}
\newcommand{\mcom}{\,,}

\newcommand{\cL}{{\cal L}}

\newcommand{\defeq}{\stackrel{\small \mathrm{def}}{=}}
\renewcommand{\leq}{\leqslant}
\renewcommand{\le}{\leqslant}
\renewcommand{\geq}{\geqslant}
\renewcommand{\ge}{\geqslant}

\newcommand\rd{\,\partial }

\newcommand{\R}{\mathbb{R}}

\usepackage{bm}
\renewcommand{\vec}[1]{{\bm{#1}}}

\newcommand{\Set}[1]{\left\{#1\right\}}

\newcommand{\poly}{{\rm poly}}

\renewcommand{\epsilon}{\varepsilon}

\newcommand{\eps}{\epsilon}

\newcommand{\remove}[1]{}

\newcommand\Rn{\ensuremath{{\mathbb{R}^n}}}
\newcommand\Rm{\ensuremath{{\mathbb{R}^m}}}

\newcommand{\Id}{\mathrm{Id}}

\newcommand\tr{\mathrm{Tr}}
\newcommand\conv{\mathrm{conv}}

\usepackage{float}

\makeatletter
\floatstyle{ruled}
\newfloat{fragment}{H}{lop}
\floatname{fragment}{Algorithm}
\renewcommand{\floatc@ruled}[2]{\vspace{2pt}{\@fs@cfont #1.\:} #2 \par
 \vspace{1pt}}
\makeatother

\title{Algorithms and Hardness for Robust Subspace
Recovery}
\author{Moritz Hardt\thanks{IBM Research Almaden. Email: {\tt mhardt@us.ibm.com}}
\and Ankur Moitra\thanks{MIT. Email: {\tt moitra@mit.edu}}}

\begin{document}

\maketitle

\begin{abstract} 
We consider a fundamental problem in unsupervised learning called \emph{subspace recovery}:
given a collection of $m$ points in $\R^n$, if many but not necessarily all of
these points are contained in a $d$-dimensional subspace $T$ can we find it?
The points contained in $T$ are called {\em inliers} and the remaining points
are {\em outliers}. This problem has received considerable attention in
computer science and in statistics. Yet efficient algorithms from computer
science are not robust to {\em adversarial} outliers, and the estimators from
robust statistics are hard to compute in high dimensions. 

Are there algorithms for subspace recovery that are both robust to outliers
and efficient?  We give an algorithm that finds $T$ when it contains more than
a $\frac{d}{n}$ fraction of the points.  Hence, for say $d = n/2$ this
estimator is both easy to compute and well-behaved when there are a constant
fraction of outliers. We prove that it is Small Set Expansion hard to find $T$
when the fraction of errors is any larger, thus giving evidence that our
estimator is an {\em optimal} compromise between efficiency and robustness. 
As it turns out, this basic problem has a surprising number of connections to other areas 
including small set expansion, matroid theory and functional analysis that we make 
use of here.
\end{abstract}

\section{Introduction}

Unsupervised learning refers to the problem of trying to find hidden structure
in unlabeled data. A ubiquitous approach is to model
this hidden structure as a low-dimensional subspace that contains many of the
data points. This approach has found a range of applications  in areas such as 
feature selection, dimensionality reduction, spectral
clustering, topic modeling and statistical inference.
There are two important desiderata for an unsupervised learning algorithm,  
\emph{computational efficiency} and \emph{robustness}: computational efficiency refers to the goal
of giving provable
guarantees on the running time of the algorithm and robustness refers to the goal of giving
guarantees that the algorithm produces a useful output even if the
assumptions of the model do not hold exactly. Our focus in this paper is
on understanding whether or not these two goals can be met simultaneously. 

Individually, these goals can each be met.  For example, there are many known
fast algorithms to compute the singular value decomposition, and from this
decomposition it is straightforward to find a low-dimensional subspace that
contains {\em all} of the data if it exists.  There are also a number of
provably robust estimators for subspace recovery.  One famous example is
{\em least median of squares} estimator of \cite{Rouss}.  The
computational problem that underlies this estimator is to find a subspace that
minimizes the median Euclidean distance to the data points.  An adversary must
corrupt at least half of the data points in order to corrupt the output.  Many
more robust estimators have been developed for this specific problem (e.g.
least trimmed squares, $M$-estimators, the Theil-Sen estimator, reweighed
least squares) and for other inference problems by the robust statistics
community (see e.g. \cite{RoussL} and \cite{Huber}).

Unfortunately, the singular value decomposition is not robust to outliers.
Moreover, only modest improvements over brute-force search are known
to actually compute the least median of squares estimator in high dimensions
(\cite{EdSouv}).
Is there an estimator for subspace recovery that is both efficiently computable and robust to
outliers? This is an instance of a fundamental and largely unexplored
question:
\begin{center}
{\slshape ``Can we reconcile computational efficiency and robustness in unsupervised
learning?''}
\end{center}
Our focus here is on a challenging notion of robustness used in the robust statistics community:
an estimator is robust if an adversary can corrupt 
an $\alpha$ fraction of the data, and the output of the estimator is still
well-behaved. The fraction of data that an adversary is allowed to corrupt is called the
{\em breakdown point} (\cite{DH}). We remark that there has been interesting recent work
on finding a subspace that approximately minimizes the sum of $\ell_p$
distances (for $p > 2$) to the data points, see, \cite{DTV,GRSW}.
Unfortunately $\ell_p$-regression can be corrupted quite easily by an adversary. 

In general, the robust statistics community studies the
breakdown properties of particular estimators.
Here, our goal is not to study a particular estimator, but rather whether or
not there is {\em any} robust estimator for subspace recovery that is also
easy to compute. The following definition is central to our paper:

\begin{definition} 
An estimator $\mathcal{E}$ is an \emph{$\alpha$-robust estimator} for 
the $d$-dimensional subspace recovery problem in $\R^n$ 
if for any set of points in 
which a $1-\alpha$ fraction are contained in a $d$-dimensional linear 
subspace $T\subseteq\R^n$, the estimator returns $T$.  
\end{definition}

Here the breakdown point is $\alpha$. So the natural question is, for what
choices of the parameters $n, d$ and $\alpha$ is there such an estimator that
is also easy to compute? There are compelling reasons to choose robust
estimators over their classical counterparts, but so far their potential has
not been realized because there are no computationally efficient algorithms to
compute them.

\subsection{Complexity of Robust Subspace Recovery}
We assume that the points outside~$T$ are in general position, and that 
the points inside~$T$ are in general position with respect to $T.$\footnote{If we remove these conditions then when $dim(T) = n-1$ the problem is equivalent to trying to satisfy as many equations as possible in an overdetermined linear system. See \cite{GR}, \cite{KM} and references therein. However, these reductions produce instances which are quite far from ones we might expect to observe in real data, and one approach for circumventing these hardness results is to instead require the above condition which is satisfied almost surely in most natural probabilistic models and seems to make the problem computationally much easier.}
Recall that the dimension of $T$ is $d$. Throughout this paper, 
we will use $L$ to denote the points inside $T$ and we will call these
the {\em inliers}, and the remaining points {\em outliers}. 
Our first result is a simple randomized algorithm
that achieves a breakdown point of $\alpha = 1- \frac{d}{n}$. 
Our result relies on Condition~\ref{cond:general}: 
any set of $n$ points is linearly independent if and only if at most $d$ of the points are inliers.

\begin{theorem}\label{thm:find}
If a set of $m$ points in $\R^n$ has strictly more than $\frac{d}{n} m$ inliers and meets Condition~\ref{cond:general}, then there is a Las Vegas algorithm
whose output is the set $L$ of inliers, each iteration can be implemented in polynomial time and the expected number of iterations is $O(n^2 m)$. 
\end{theorem}

\noindent In fact, an interesting comparison can be drawn between our
algorithm and the famous {\em RANSAC} method of \cite{FB}:
Both approaches repeatedly select a random set of $n$ points; {\em RANSAC}
works when this sample contains {\em only} inliers, whereas our algorithm
works when the sample contains at least $d+1$ inliers. The main observation is
that even if a set of $n$ points contains many outliers, only the inliers can
participate in a linear dependence. In fact, for $d = n/2$ and even if inliers
make up $3/4$ of the points, {\em RANSAC} will take an exponential number of
iterations to find $T$ while our algorithm requires only a constant number of
iterations (see Remark~\ref{remark:const}). 

Our algorithm can also be made stable in that the inliers do not
need to be exactly contained within $T$.  Here we need
Condition~\ref{cond:general2}: the smallest determinant of any set of points
with at most $d$ inliers is strictly larger than the largest determinant of
any set of points with at least $d+1$ inliers. 

\begin{theorem}\label{thm:find2}
If a set of $m$ points in $\R^n$ has strictly more than $\frac{d}{n} m$ inliers and meets Condition~\ref{cond:general2}, then there is a Las Vegas algorithm
whose output is the set $L$ of inliers, each iteration can be implemented in polynomial time and the expected number of iterations is $O(n^2 m)$. 
\end{theorem}

Our estimator achieves a constant breakdown point for, say, $d = n/2$. 
Yet there are numerous inefficient estimators that achieve a better breakdown point (e.g. a constant breakdown point
even when $d = n-1$). We provide evidence that our estimator is the {\em optimal} compromise between
efficiency and robustness: it is {\em small set expansion} hard to improve
the breakdown point beyond this threshold. We state our result informally here:



\begin{theorem}
There is an efficient reduction from an instance of $(\epsilon, \delta)$-{\sc Gap-Small-Set Expansion} on a graph $G$ to {\sc
Gap Inlier} such that:
\begin{itemize}
\item if there is a small non-expanding cut in $G$ then there exists a subspace of dimension 
$d$ containing at least $(1-\epsilon)\frac{d}{n}$ fraction of the points

\item and if there is no small non-expanding cut
then every subspace of dimension $d$ contains at most a $2\epsilon\frac{d}{n}$
fraction of the points.
\end{itemize}
\end{theorem}

\noindent \cite{K} proved a related result that it is NP-hard to
find a $d = n-1$ dimensional subspace that contains a $(1-\epsilon)\frac{n-1}{n}$ fraction
of the data points\footnote{This follows by applying a padding argument to the knapsack instance before 
proceeding with the reduction in \cite{K}.}. In general, it seems difficult to base hardness for robust subspace recovery
(when $d < n-1$) on standard assumptions and this is an interesting open question. 

%

Taking a step back, computational complexity is an important lens for understanding learning and
statistical problems in the sense that there are many sample-efficient
estimators, e.g., maximum likelihood, that are hard to compute, but by allowing
more samples than the information theoretic minimum we can find alternatives
that are easy to compute. Yet these hard estimators are still favored in
practice, perhaps not just due to their sample efficiency but also due
to their robustness. A broader goal of our paper is to bring to light
questions about whether there are estimators that meet all three objectives of
being efficiently computable, sample efficient and robust (and not just two
out of three). 


\subsection{Derandomization and Duality for Robust Subspace Recovery}

The crucial step in our randomized algorithm is to repeatedly sample subsets 
of $n$ points and once we find one that is linearly dependent, we can use this 
subset to recover the set of inliers. If a collection of $m$ points in $\R^n$ has the 
property that a random subset of $n$ points is linearly dependent (with non-negligible probability), 
can we find such a subset deterministically? We give a solution to this problem
using tools from matroid theory.

Indeed, a well-studied polytope in matroid literature is the {\em basis
polytope} which is the convex hull of all sets of $n$ points that form a basis
(see Section~\ref{sec:radbp}). Condition~\ref{cond:general} guarantees us that
the vector $\frac{n}{m} \vec{1}$ is outside the basis polytope, and our goal
of finding a set of $n$ points that do not span $\R^n$ can be stated
equivalently as finding a Boolean vector (whose coordinates sum to $n$) that
is also outside the basis polytope.

There has been a vast literature on the basis polytope and on submodular minimization,
and there are deterministic strongly polynomial time algorithms for deciding membership
in the basis polytope---see \cite{E,C,GLS,S,IFF}.
Our idea is in each step we find a line segment $\ell$ that contains the
current vector (starting with $\frac{n}{m} \vec{1}$). Since the current vector
is outside the basis polytope it is easy to see that at least one of the
endpoints of $\ell$ must also be outside. So we can move the current vector to
this endpoint and if we choose these segments $\ell$ in an appropriate way we
will quickly find a Boolean solution. The key is that a membership oracle for
the basis polytope tells us which endpoint of $\ell$ we should move to. Hence
we obtain an algorithm that is not only an optimal tradeoff between efficiency
and robustness, but is even deterministic: 
\begin{theorem}
If a set of $m$ points in $\R^n$ has strictly more than $\frac{d}{n} m$ inliers and meets Condition~\ref{cond:general}, then there is a deterministic polynomial time algorithm
whose output is the set $L$ of inliers.
\end{theorem}
The basis polytope not only plays a central role in robust subspace recovery but is also closely related to a notion studied in functional
analysis that we call \emph{radial isotropic position}. 
In fact, \cite{Barthe98} studied a convex programming
problem whose optimal solution finds a linear transformation that places a set
of points in radial isotropic position (see Section~\ref{sec:barthe}) if it
exists. The connection is that the optimal value to this convex program is finite (i.e. there is
such a transformation) if and only if the vector $\frac{n}{m} \vec{1}$ is
inside the basis polytope. 

Barthe's convex program provides a connection between radial isotropic
position and robust subspace recovery: just as placing a set of points in
isotropic position is a proof that the set of points is not contained in a
low-dimensional subspace, so is placing a set of points in radial
isotropic position a proof that there is no $d$-dimensional subspace that
contains more than a $\frac{d}{n}$ fraction of the points (see
Section~\ref{sec:radbp}). We give effective bounds on the region in which an optimal
solution to the convex program is contained, and how strictly convex the function is and use this to give an
efficient algorithm to compute radial isotropic position. 

\begin{theorem}[informal]
There is a deterministic polynomial time algorithm to compute a linear
transformation $R$ that places a set of points in radial isotropic position,
if such a transformation exists. 
\end{theorem}

\noindent Notably, this theorem shows that if there is no low-dimensional subspace that contains many 
of the points, we can deterministically compute a certificate that there is no such subspace. 

Radial isotropic position can also be thought of as a more stable analogue
of isotropic position that is not sensitive to either the norms of the data
points or to a constant fraction of adversarial outliers! Isotropic position
has important applications both in algorithms and in exploratory data
analysis, but is quite sensitive to even a small number of outliers (see e.g.
\cite{Vem}). 
%
Just as robust statistics asks for estimators that are well-behaved in the presence of outliers,
we could ask for {\em canonical forms }(e.g. isotropic position, radial isotropic position) that are well-behaved
in the presence of outliers.
Perhaps radial isotropic position will be a preferable
alternative in some existing applications where being robust is crucial.

Somewhat surprisingly, this elementary problem of finding a low-dimensional subspace that contains many
of the data points is connected to a number of problems and combinatorial objects including the small set expansion hypothesis, the independent set polytope and submodular minimization 
and notions in functional analysis, and we make use of all of these connections.

\subsection{Related Work}

Our work fits into a broader agenda within statistics and machine learning:
Can we recover a low-rank matrix from noisy or incomplete observations? The
foundational work of \cite{RFP} and \cite{CR} gave convex programming
algorithms that provable recover a low-rank matrix when given a small number
of random chosen entries in the matrix. These techniques have since been
adapted to settings in which an adversary can corrupt some of the entries in
the matrix \cite{CSPW,PCA,XCS}. However we note that there are two
incomparable models for how an adversary is allowed to corrupt the entries in
a low-rank matrix, and which model is more natural depends on the setting. For
example, the exciting work of \cite{PCA} considers a model in which an
adversary can corrupt a constant fraction of the entries of $A$ whose
locations are chosen uniformly at random. In contrast, the model in
\cite{XCS,ZL} for example allows an adversary to corrupt a large
fraction of the columns of $A$. This is the setting in our work, and this
assumption is most natural when we think of columns of $A$ as representing
individuals from a population and uncorrupted columns correspond to
individuals that fit the model, but we would like to make as few assumptions
as possible about the remaining individuals that do not fit the model. The
breakdown point of robust subspace recovery was also studied recently by
\cite{XCM13} where successful subspace recovery is guaranteed in a suitable
stochastic model.

We note that much of the recent work from statistics and machine learning has
focused on a setting where one posits a distributional model that generates
both the inliers and outliers and the goal is to recover the subspace $T$ with
high probability. For example, see the recent work of \cite{SC} and
\cite{Needle} and references therein.  In principle, our work is not directly
comparable to these models since our results are not contingent on any one
distributional model.  Yet in some of these probabilistic models (e.g. in
\cite{Needle}) the probability that a point is chosen from the subspace $T$ is
larger than $\frac{d}{n}$ in which case Condition~\ref{cond:general} is
satisfied with high probability and hence our algorithm succeeds in these
cases too. 

The above discussion has focused on notions of robustness that allow an adversary to corrupt a constant fraction of the entries in 
the matrix $A$. However, this is only one possible definition of what it means for an estimator to be robust to noise. For example,
principal component analysis can be seen as finding a $d$-dimensional subspace that minimizes the sum of squared distances to the
data points. A number of works have proposed modifications to this objective function (along with approximation algorithms) in the hopes
that this objective function is more robust. As an example,  \cite{DTV} gave
a $O(p^{p/2})$ approximation algorithm for the problem of finding a subspace that minimizes the sum of $\ell_p$ distances to the data points (for $p > 2$). Another
example is the recent work of  \cite{NRV} which gives a constant factor approximation for finding a $d$-dimensional subspace that
maximizes the sum of Euclidean lengths of the projections of the data points (instead of the sum of squared lengths). Lastly, we mention that
 \cite{DunaganV00} gave a geometric definition of an outlier (that does not depend on a hidden subspace $T$) and give an optimal algorithm for removing outliers according to this definition.

\subsection*{Acknowledgments}
We thank Quentin Berthet for pointing out an issue in an earlier version of
this manuscript. Thanks to Gilad Lerman for many helpful discussions and to
Joel Tropp for pointers to the literature.

\section{A Simple Randomized Algorithm}\label{sec:simple}

Here we give a randomized algorithm for robust subspace recovery. The idea is that once we find any 
non-trivially sparse linear dependence we can use it to find the set of inliers provided that the inliers are
in general position with respect to $T$. The breakdown point of this estimator is exactly the threshold at which
a random set of $n$ points is linearly dependent with non-negligible probability. Surprisingly, in Section~\ref{sec:sse}
we give evidence based on the small set expansion conjecture that there is no efficient estimator that has a better
breakdown point. 

We will think of an instance of robust subspace recovery as a matrix
$A\in\R^{n\times m}$ with $m\ge n$ and rank $n.$ Throughout this paper for $V \subset [m]$, we will let $A_V$ denote the 
submatrix corresponding to columns in $V$. Suppose that there is a
$d$-dimensional subspace $T$ that contains {\em strictly} more than a
$\frac{d}{n}$ fraction of the columns of $A$. Our goal is to recover this
subspace (under mild general position conditions on these points) efficiently.
 Let $L \subset [m]$ be the columns of
$A$ that are inliers. We will need the following condition which is almost surely satisfied by any reasonable probabilistic model
that generates inliers from the subspace $T$ and outliers from all of
$\R^{n}$:

\begin{condition}\label{cond:general}
A set of $n$ columns of $A$ is linearly independent if and only if 
at most $d$ of the columns are inliers.
\end{condition}

\begin{fragment*}[ht]
\caption{
\label{alg:rand}{\sc RandomizedFind}\\
\textbf{Input: } $A\in\R^{n\times m}$ which satisfies Condition~\ref{cond:general} \vspace*{0.01in}
}

\begin{enumerate} \itemsep 0pt
\small 
\item $start:$ Choose $V \subset [m]$ with $|V| = n$ uniformly at random
\item If $rank(A_V) < n$, 
\item $\quad$ Find $u \in ker(A_V)$; Set $\cL = span(\{A_i \colon u_i \neq 0
\})$; Set $L = \{i \colon A_i \in \cL\}$
\item $\quad$ Output $L$
\item Else return to \emph{start}
\end{enumerate} 

\end{fragment*}

The next lemma gives a lower bound on the
probability of sampling strictly more than expected number of inliers:
\begin{lemma}\label{lemma:lasvegas}
Suppose that we are given a set of $m$ points in $\R^n$ with strictly more
than $\frac{d}{n} m$ inliers. Let $V$ be a uniformly random set of $n$ points
(without repetition). Then the probability that $U$ contains at least $d + 1$
inliers is at least $p \geq \frac{1}{2n^2m}$.  
\end{lemma}

\begin{proof}
Let $X$ be a random variable defined to be the number of inliers in a random
set $V$ of $n$ points. Then $E[X] > d$ and set $\widehat{X} = X - E[X]$. Then
let $p$ be the probability that $\widehat{X} \geq 0$, and this condition
certainly implies that we have at least $d + 1$ inliers. Since the expectation
of $\widehat{X}$ is zero, we have that $$ p E[\widehat{X} | \widehat{X} \geq 0] + (1-p)
E[\widehat{X} | \widehat{X} < 0] = 0$$ Then we can upper bound $E[\widehat{X}
| \widehat{X} \geq
0] \leq n-d$ and $-p E[\widehat{X} | \widehat{X} < 0] \leq p n$. 
Hence, 
$$ p (n-d) + pn
\geq -E[\widehat{X} | \widehat{X} < 0] \geq \frac{\lfloor \frac{d}{n} m \rfloor + 1}{m}
- \frac{d}{n} \geq \frac{1}{nm}$$ 
and this completes the proof of the lemma.
\end{proof}

\begin{remark}\label{remark:const}
We remark that the lower bound on $p$ can be improved to $p\ge(d/n)^2/2$
when $m\ge 6n+2$ and $n\ge 3.$ Hence our algorithm is quite practical in this
range of parameters. 
\end{remark}
Indeed, with the same notation as above, condition $X$ on the
event $E$ that the first two samples are contained in $L$ (the set of
inliers). Clearly, $\Pr\Set{E}\ge (d/n)^2.$ On the other hand, we still sample
$n-2$ points with replacement. Each sample now has a probability of landing in
$L$ that is at least
$q \ge \frac{dm/n - 2}{m-2} = 
\frac dn - \frac{n+d}{n(m-2)}
\ge \frac dn - \frac1{3n}\mper$
Here, we used that $m\ge 6n + 2.$ Hence,
$\E[X\mid E]\ge (n-2)(\frac dn - \frac1{3n})
\ge d-1\mcom$
where we used that $2/n + 1/3\le 1.$ On the other hand, we have
$\Pr\Set{X\ge \lfloor\E[X\mid E]\rfloor \mid E}\ge \frac12$
by the ``mean is median'' theorem for hypergeometric distributions
(see, e.g., \cite{KI}). It follows that 
$\Pr\Set{X \ge d+1} \ge 
\Pr\Set{X \ge d-1\mid E}\Pr\Set{E}
\ge \frac12 (d/n)^2\mper$

 The next claim captures the intuition that from any non-trivially sparse
linear dependence in $A$ it is easy to compute the set of inliers.

\begin{claim}\label{claim:support}
If $|V| = n$, then any vector in the kernel of $A_V$ must contain $d+1$ inliers in its support and no outliers.
\end{claim}

\begin{proof}
Suppose there is a vector $u \in ker(A_V)$ with $u_i \neq 0$ for $i \notin L$.
Since any $d+1$ inliers are linearly dependent, we can use Caratheodory's
Theorem (see \cite{Matousek}) to find a vector $v \in ker(A_V)$ supported on
at most $d$ inliers and for which $v_i \neq 0$. This contradicts
Condition~\ref{cond:general} since the support of any non-zero vector in the
kernel must contain at least $d + 1$ inliers. 
\end{proof}

\begin{fragment*}[t]
\caption{
\label{alg:reduce}{\sc RandomizedFind2}\\
\textbf{Input: } $A\in\R^{n\times m}$ which satisfies Condition~\ref{cond:general2} \vspace*{0.01in}
}

\begin{enumerate} \itemsep 0pt
\small 
\item $start:$ Choose $V \subset[m]$ with $|V| = n$ uniformly at random
\item If $det(A_V^T A_V) < C^2$
\item $\qquad$ While $|V| > d + 1$
\item $\qquad$ $\qquad$ Find $\{u\}$ such that $det(A_{V - \{u\}}^T A_{V -
\{u\}}) < C^2$; Set $V = V - \{u\}$
\item $\qquad$ Set $\cL =V \cup \{v \mid det(A_{V\Delta \{u, v\}}^T A_{V\Delta \{u, v\}}) < C^2\} $ where $u \in V$
\item Else return to \emph{start}.
\end{enumerate} 

\end{fragment*}

%
\begin{proof}[Theorem~\ref{thm:find}]
Claim~\ref{claim:support} guarantees the correctness of the algorithm, and
Lemma~\ref{lemma:lasvegas} guarantees that the success probability of each
iteration is at least $p \geq \frac{1}{2n^2m}$ and this implies the lemma.
\end{proof}

We are interested in generalizing our algorithm to the setting where inliers are only
approximately contained in the subspace. This idea is formalized next. 

\begin{condition}\label{cond:general2}
Any set $V$ of at most $n$ columns of $A$ has $det(A_V^T A_V)\ge C^2$ if the number of inliers is at most $d$, and otherwise strictly less than $C^2$. 
\end{condition}

We can now prove Theorem~\ref{thm:find2}, a stable analogue of Theorem~\ref{thm:find}.

\begin{proof}
Lemma~\ref{lemma:lasvegas} guarantees that the probability that the algorithm finds a set $V$ with $|V| = n$ and $det(A_V) < C$ in  is at least $p \geq \frac{1}{2n^2m}$, and furthermore the algorithm maintains the invariant that the set $V$ always has at least $d + 1$ inliers, and at the end of the while loop $V$ is a set of $d + 1$ inliers. Then Condition~\ref{cond:general2} guarantees that the algorithm correctly outputs the set of inliers. 
\end{proof}

\section{Computational Limits}\label{sec:sse}

We will now present evidence that the robust subspace recovery problem is
computationally hard beyond the breakdown point achieved by our randomized
algorithm in Section~\ref{sec:simple}. For this purpose we need to introduce the \emph{expansion
profile} of a graph.
Given a $\Delta$-regular graph $G=(V,E)$ we define the edge expansion of
a set $S\subseteq V,$ as
\[
\phi_G(S)=\frac{|E_G(S,V\backslash S)|}{\Delta|S|}\mper
\]
Here and in the following, we let $E_G(A,B)$
denote the set of edges in $G$ with one endpoint in $A$ and
the other in~$B.$ Let us also denote $\mu(S)=|S|/|V|.$
Given a parameter $\delta\in[0,1/2],$ we define the \emph{expansion profile}
of $G$ as the curve
\[
\phi_G(\delta) = \min_{\mu(S)=\delta}\phi(S)\mper
\]
With these definitions we describe the Small Set Expansion problem as was
recently studied by~\cite{RaghavendraS10,RaghavendraST10}:
\begin{definition}
The {\sc Gap-Small-Set Expansion} problem is defined as: Given a graph $G,$
and constants $\epsilon,\delta>0,$ distinguish the two cases
(1) $\phi_G(\delta)\ge 1-\epsilon,$ and, (2)
$\phi_G(\delta)\le \epsilon.$
\end{definition}
We will relate the previous problem to the {\sc Gap-Inlier} problem that we
define next.
\begin{definition}
The {\sc Gap-Inlier} problem is defined as: Given $m$ points
$u_1,\dots,u_m\in\R^n,$ and constants $\epsilon,\delta,$ distinguish the two cases
\begin{enumerate}
\item there exists a subspace of dimension $\delta n$ containing a
$(1-\epsilon)\delta$ fraction of the points,
\item every subspace of dimension $\delta n$ contains at most a
$\epsilon\delta$ fraction of the points.
\end{enumerate}
\end{definition}

Our next theorem shows a reduction from {\sc Gap-Small-Set Expansion} to {\sc
Gap Inlier}.

\begin{theorem}
Let $\epsilon,\delta>0.$
There is an efficient reduction which given a $\Delta$-regular graph
$G=(V,E),$ produces an instance $u_1,\dots,u_m\in\R^n$ of {\sc Gap-Inlier}
such that
\begin{description}
\item[Completeness:]
If $\phi_G(\delta)\le\epsilon,$ then there exists a subspace of dimension 
$\delta n$ containing at least $(1-\epsilon)\delta$ fraction of the points.
\item[Soundness:]
If $\phi_G(\delta')\ge1-\epsilon$ for every $\delta'\in[\delta \epsilon/2,2\delta],$ then
every subspace of dimension $\delta n$ contains at most a $2\epsilon\delta$
fraction of the points.
\end{description}
\end{theorem}

\begin{proof}
Our reduction works as follows.
Let $G=(V,E)$ be an instance of {\sc Gap-Small-Set Expansion}.
Let $m=|E|$ and $n=|V|.$ 
For each edge $e=(i,j),$ create a vector $u_e= \alpha_e e_i + \beta_e e_j,$
where $e_i$ is the $i$-th standard basis vector and $\alpha_e,\beta_e$ are
drawn independently and uniformly at random from~$[0,1].$ This defines an
instance $u_1,\dots,u_m\in\R^n$ of {\sc Gap-Inliers}.

To analyze our reduction, it
will be helpful to consider the following intermediate graph. Let $B=(E,V)$ be
the bipartite graph where we connect each edge $e\in E$ with the two vertices
in~$V$ that it is incident to. Note that $B$ is $(2,\Delta)$-regular. The next
claim relates the dimension of a set of points to the size of the neighborhood
of the corresponding edge set in the graph~$B.$ Given a set $F \subseteq E$ we will
use $N_B(F)$ to denote the neighborhood of $F$. 

\begin{claim}
\label{dim-points}
For every set of points $P\subseteq\{u_1,\dots,u_m\}$ corresponding to a set of
edges $F\subseteq E$ we have with probability~$1$ over the choice of the
coefficients above
\[
\frac{|N_B(F)|}{2}  \leq \dim\left (\mathrm{span}(P)\right) \leq \left| N_B(F) \right|\mper
\]
\end{claim}
\begin{proof}
On the one hand, the points $P$ are contained in the
coordinate subspace of dimension~$d=|N_B(F)|$ corresponding to the union of the
support of the vectors. On the other hand, suppose the set $F$ is a spanning tree. Then
the vector corresponding to each edge is linearly independent of the others, and hence
the dimension of the subspace is precisely $|N_B(F)| - 1$ almost surely. Similarly, if $F$ corresponds to
$r$ connected components, then the dimension of the subspace is $|N_B(F)| - r$ almost surely, and
the number of connected components is at most $\frac{|N_B(F)|}{2}$ since there are no isolated vertices.
This implies the claim.
\end{proof}

\paragraph{Completeness.}
We begin with the completeness claim. Let $S\subseteq V$ be a set of measure~$\delta$ 
and suppose that $\phi_G(S)\le\epsilon.$ Double-counting the
edges spanned by $S,$ we get
\[
\Delta|S| = 2 |E_G(S,S)| + |E_G(S,V\backslash S)|\mper
\]
Hence, $|E(S,S)| \ge \Delta|S|/2 - \epsilon \Delta|S|/2 = (1-\epsilon)\Delta|S|/2.$ On the
other hand the edge set $E_G(S,S)$ has at most $|S|$ neighbors in $B.$ This
implies that the points corresponding to $E(S,S)$ are contained in a
coordinate subspace of dimension $|S|.$ Equivalently, there exists a $(\delta
n)$-dimensional subspace containing at least $(1-\epsilon)\delta\Delta n/2$
points. Since $m=\Delta n/2,$ this corresponds to a fraction of~$(1-\epsilon)\delta$
which is what we wanted to show.

\paragraph{Soundness.}
Next we establish soundness. Consider any set of points~$P$ contained
in $\delta n$ dimensions. We will show that under the given assumption on the
expansion profile of~$G,$ it follows that  $|P|\le \epsilon\Delta\delta
n.$ Again, since $m=\Delta n/2,$ this directly implies that any subspace of
dimension $\delta n$ contains at most a $2\epsilon\delta$ fraction of the
points. Let $F$ be the set of edges corresponding to $P$ in the graph $B$ and
let~$S$ be its vertex neighborhood in $B.$
Suppose (for the sake of contradiction) that
$|F| > \epsilon\Delta \delta n.$
First, note that the neighbor set $S\subseteq V$ of 
$F$ in the graph $B$ satisfies
\begin{equation}\label{size-S}
\frac{\delta \epsilon n}{2}\le |S| \le 2\delta n\mper
\end{equation}
The second inequality follows from Claim~\ref{dim-points}. The first
inequality follows because $G$ is $\Delta$-regular. Thus, a set of $|S|$
vertices can induce at most $\Delta|S|/2$ edges and all edges in $F$ are
induced by $S.$

Counting the edges touching $S$ as before, 
\[
\Delta|S|
= 2|E_G(S,S)| + |E_G(S,V\backslash S)| 
\ge 2|E_G(S,S)| + \Delta(1-\epsilon)|S|\mper
\]
The inequality followed from our assumption on the expansion profile of~$G$
which we may apply because $S$ satisfies Equation~\ref{size-S}.
Consequently:
\[
|E_G(S,S)|\le \frac{\epsilon\Delta|S|}{2}\mper
\]
On the other hand, $|F|\le|E_G(S,S)|,$ since every edge in $F$ is induced by
$S.$ Hence, the previous inequality showed that 
$|F| \leq \epsilon\Delta\delta n/2.$ and this is a contradiction, which completes the proof.
\end{proof}

\section{The Basis Polytope}\label{sec:radbp}

Here we connect the {\em independent set polytope} which has received considerable attention in matroid literature, to
a notion studied in functional analysis that we call {\em radial isotropic position}. In Section~\ref{sec:clust} we will use known
algorithms for deciding membership in the independent set polytope to derandomize our algorithm from Section~\ref{sec:simple}. 
And in Section~\ref{sec:computerip} we will give an efficient algorithm to compute radial isotropic position, which can be thought of
as a robust analogue to isotropic position. Let $A = [u_1,\dots,u_m]\in\R^{n\times m}$ with $m\ge n$. 

\begin{definition}
Let $P$ be the {\em independent set polytope} defined as:
$$P \defeq \conv \Big \{ \vec{1}_U \colon U\subseteq[m], 
\dim\big(\mathrm{span}\Set{u_i\colon i\in U}\big)=|U| \Big \} \mcom$$
where $\mathbf{1}_U$ is the $m$-dimensional indicator vector of the set $U.$ 
Also let $K_A$ be the {\em basis polytope} which is
the facet of $P$ corresponding to $\sum_i x_i = n$. 
\end{definition}

These polytopes can be defined (in a more general context) using the language of matroid theory 
where independent sets of vectors are replaced by independent sets in a matroid. 
A fundamental algorithmic problem in matroid theory is to give an efficient membership oracle
for these polytopes. A number of solutions are known which all follow from a characterization of Edmonds \cite{E}
that reduces membership to solving a submodular minimization problem:
$\min_{U \subset [m]} \mathrm{rank}( \{u_i \colon i \in U \}) - \sum_{i \in U}
x_i\mper$
The optimum value of this minimization is nonnegative if and only if $x \in P$
(\cite{E}). Hence an immediate consequence
of the known algorithms for submodular minimization \cite{GLS}, \cite{S}, \cite{IFF} and even a direct algorithm of \cite{C} yield:

\begin{theorem}
There is a deterministic polynomial time algorithm to solve the membership problem for the independent set polytope $P$ (and the basis polytope $K_A$). 
\end{theorem}

We will use this tool from matroid theory to derandomize our algorithm from Section~\ref{sec:simple}. Recall that the main step in our algorithm
is to repeatedly sample subsets of $n$ points and once we find one that is linearly dependent, we can use this subset to recover the set of inliers. 
So our approach is to use a membership oracle for the basis polytope to find a subset of $n$ points that is linearly dependent deterministically. 

The basis polytope not only plays a central role in robust subspace recovery but also in a notion studied in functional analysis called radial
isotropic position. These two concepts can be thought of as dual to each other: 
Recall that the set of vectors $u_1,\dots,u_m\in\R^n$ is in {\em isotropic position} if
\[\textstyle
\sum_{j=1}^m u_j\otimes u_j = \Id_n\mper
\]
 It is well-known that a set of points can be placed in isotropic position if and only if the points are not
all contained in an $n-1$-dimensional subspace. Just as isotropic position can be thought of as a certificate that a set of points is full-dimensional, so
too radial isotropic position can be thought of as a certificate that there is no low-dimensional subspace that contains many of the points.

\begin{definition}
We say that a linear transformation $R\colon\R^n\to\R^n$ puts set of vectors $u_1,\dots,u_m\in\R^n$ in \emph{radial isotropic
position} with respect to a coefficient vector 
$c\in\R^m$ if $\sum_{i=1}^m c_i \frac{Ru_i}{\|Ru_i\|} 
\otimes\frac{Ru_i}{\|Ru_i\|} = \mathrm{Id}_{n}.$
\end{definition}

If a set of vectors meets Condition~\ref{cond:general} then it cannot be put in radial isotropic position: any linear transformation $A$ preserves the invariant
that the inliers lie in a subspace of dimension $d$, but after applying $A$ and rescaling the points to be unit vectors the variance of a random sample restricted to this
subspace is strictly larger than $d$, which is too large! More generally, when can a set of vectors be put in radial isotropic position? \cite{Barthe98} gave a complete answer to this question: 




\begin{theorem}[Barthe]
A set of vectors $u_1,\dots,u_m\in\R^n$ can be put in \emph{radial isotropic position} 
with respect to
$c\in\R^m$ if and only if $c\in K_A$. 
Moreover, $c\in K_A$ if and only if the following supremum has finite value:
$\sup_{t\in\R^m} \langle c,t\rangle - \log\det\left(\sum_{i=1}^m
e^{t_i}u_i\otimes u_i\right).$
\end{theorem}

This concave maximization problem provides a connection between radial
isotropic position and robust subspace recovery.  The optimal value reveals to
us which case we are in: if it is finite, then the points can be put in radial
isotropic position but if it is infinite (under Condition~\ref{cond:general})
then there is a subspace $T$ of dimension $d$ that contains more than a
$\frac{d}{n}$ fraction of the points!

\section{A Deterministic Algorithm}\label{sec:clust}

In this section we apply tools from matroid theory (see \cite{E}, \cite{C}) to derandomize our algorithm from Section~\ref{sec:simple}. 
Recall that the main step in our algorithm from Section~\ref{sec:simple}
is to repeatedly sample subsets of $n$ points and once we find one that is linearly dependent, we can use this subset to recover the set of inliers. 
Our goal is to find such a subset deterministically, and we can think about this problem instead in terms of the basis polytope. 

Condition~\ref{cond:general} guarantees that the vector $\frac{n}{m} \vec{1}$ is outside the basis polytope. We remark that a set of $n$ columns 
is linearly dependent if and only if the indicator vector is outside the basis polytope. So we can think about this derandomization problem instead
as a rounding problem: we are given a vector $\frac{n}{m} \vec{1}$ that is outside the basis polytope and we would like to round it to a Boolean
vector (that sums to $n$) that is also outside the basis polytope. 

Our approach is simple to describe, and builds on known polynomial time membership oracles for the basis polytope developed within 
combinatorial optimization (see, \cite{E}, \cite{C}, \cite{GLS}, \cite{S},
\cite{IFF}). In each step
we find a line segment $\ell$ that contains the current vector (starting with $\frac{n}{m} \vec{1}$). Since the current vector is outside
the basis polytope it is easy to see that at least one of the
 endpoints of $\ell$ must also be outside. So we can move the current vector to this endpoint
 and if we choose these segments $\ell$ in an appropriate way we will quickly find a Boolean solution. 

Indeed, Edmonds gave a general characterization of the independent set polytope:
\begin{theorem}[\cite{E}]
The independent set polytope $P$ can equivalently be described as $P \defeq
\Big \{ x \in \R^m \colon \mbox{ for all } U \subset [m],
\dim\big(\mathrm{span}\Set{u_i\colon i\in U}\big) \geq \sum_{i \in U} x_i \Big
\}.$
\end{theorem}

Hence we can intersect this alternative description of $P$ with the constraint $\sum_i x_i = n$ to obtain an alternative description of the basis polytope that will be more
convenient for our purposes. Indeed, if Condition~\ref{cond:general} is met then any subset $U$ of points has rank equal to $\min(n, \min(|U \cap L|, d) + |U / L|)$ and so:

\begin{corollary}\label{lemma:poly}
If a set of $m \geq n$ points meets Condition~\ref{cond:general}, then 
\[\textstyle
P = \Big \{x \in \R^{m} \colon 0 \leq x_i \leq 1,  \sum_{i=1}^m x_i \leq n \mbox{ and }
\sum_{i \in L} x_i \leq d \Big \}
\]
\[\textstyle
K_A = \Big \{ x \in \R^{m} \colon 0 \leq x_i \leq 1, \sum_{i=1}^m x_i = n \mbox{ and }
\sum_{i \in L} x_i \leq d \Big \}
\]
\end{corollary}

\begin{fragment*}[t]
\caption{
\label{alg:reduce}{\sc DerandomizedFind}\\
\textbf{Input: } $A\in\R^{n\times m}$ which satisfies Condition~\ref{cond:general} \vspace*{0.01in}
}

\begin{enumerate} \itemsep 0pt
\small 
\item Set $U = [m]$
\item While $|U| > n $
\item $\qquad$ For each $i \in U$
\item $\qquad$ $\qquad$ Check if $\frac{n}{|U \backslash \{i\}|} \vec{1} \in
K_{A_{U \backslash \{i\}}}$
\item $\qquad$ $\qquad$ If `NO', Set $U = U \backslash \{i\}$ (exit for loop)
\item Find $u \in ker(A_U)$, Set $T = span(\{A_i \colon u_i \neq 0 \})$, Return $L = \{i \colon A_i \in T\}$
\end{enumerate} 

\end{fragment*}

\begin{lemma}\label{lemma:nomult}
After exiting the while loop, $|U \cap L | \geq d + 1$
\end{lemma}

\begin{proof}
An immediate consequence of Corollary~\ref{lemma:poly} is that for each call to the membership oracle for $K_V$ for some set $V$, the answer is `NO' if and only if the fraction of inliers in $V$ is more than $\frac{d}{n}$. Since at the start of the while loop we are guaranteed that the fraction of inliers in $U$ is more than $\frac{d}{n}$, this is an invariant of the algorithm. All that remains is to check that for any set $U$ with $|U| > n$ and more than a $\frac{d}{n}$ fraction of inliers, there some element $i$ that we can remove from $U$ to maintain this condition (i.e. the algorithm does not get stuck). This is easy to check since if $U$ contains even just one outlier, we can choose $i$ to be that element and this will only increase the fraction of inliers and if instead there are no outliers left then we can choose any inlier to remove. Hence the algorithm does not get stuck, outputs a set $U$ with $|U| = n$ which has strictly more than a $\frac{d}{n}$ fraction of inliers and so $|U \cap L | \geq d + 1$. 
\end{proof}

\begin{theorem}
Given a set of $m$ points $u_1, \dots,u_m \in\R^n$ with $m \geq n$ that meets Condition~\ref{cond:general} and which $T$ contains more than a $\frac{d}{n}$ fraction of the points, then 
 {\sc DerandomizedFind} computes $T$. The running time of this algorithm is bounded by a fixed polynomial in $n$, $m$.
\end{theorem}

\begin{proof}
Since $|U \cup L| \geq d + 1$, we have that $rank(A_U) < n$. Then using Claim~\ref{claim:support}, {\sc DerandomizedFind} computes the span $T$ of the inliers, and outputs exactly the set of inliers. Note that there are a number of known strongly polynomial time algorithms for deciding membership in $K_A$ (see Section~\ref{sec:radbp}). 
\end{proof}

\bibliographystyle{moritz}
\bibliography{robust}

\appendix

\section{Barthe's Convex Program}\label{sec:barthe}
Recall that the basis polytope $K_A$ characterizes exactly when we can put
a set of points in radial isotropic position (\cite{Barthe98}). There are several
known algorithms from the matroid literature that provide a strongly polynomial time
algorithm for deciding membership in $K_A.$ However, the focus of this section and the
next is not just deciding if the optimization problem of Barthe has finite or infinite value,
but finding an optimal solution in case that the optimum is finite. From the solution to this
optimization problem, we will be able to derive the linear transformation that
places a set of points in radial isotropic position. 
Here we will explain in detail the connection found by \cite{Barthe98} and 
\cite{CarlenLiLo04,CarlenCo08} between convex programming and radial isotropic position. In the 
next section we will prove various effective bounds on this convex programming problem that
we need in order to show that the Ellipsoid method finds an optimal solution. 


Recall that Barthe considers maximizing a concave function (or equivalently minimizing a convex function):
\[
\sup_{t_1,\dots,t_m\in\R} \langle c,t\rangle - \log\det\left(\sum_{i=1}^m
e^{t_i}u_i\otimes u_i\right)
\]
for a given set of points $u_1,\dots,u_m\in\R^n$ and a coefficient vector
$c\in\R^m.$
How is this unconstrained maximization problem related to the linear
transformation that puts the points $u_1,\dots,u_m$ into radial isotropic
position?  For now we specialize our discussion to the case in which $c =
\frac{n}{m}\vec{1}$, where $\mathbf{1}$ is the all ones vector. Let
$t_1,\dots,t_m\in\R.$ Consider the matrix $U=\sum_{j=1}^me^{t_j}u_j\otimes
u_j.$ We know that this matrix is positive definite and has full rank.
Therefore it has a symmetric positive definite square root and we can define
$R=U^{-1/2}.$
Notice that 
$$\Id_n=U^{-1/2}UU^{-1/2}=\sum_{j=1}^m e^{t_j} Ru_j\otimes Ru_j$$
Hence, we have what we need if we can choose $t_j$ such that 
$e^{t_j} = \frac nm \|Ru_j\|^{-2}\mper$
The crucial insight is that these conditions are exactly the optimality conditions
in Barthe's maximization problem. 

\begin{lemma}[\cite{Barthe98}]\label{lem:opt}
Let $A=[u_1,\dots,u_m]$ denote a matrix with column vectors
$u_1,\dots,u_m\in\Rn.$ Suppose $\phi_A^*(c)<\infty.$ 
Then, any optimal solution $t_1,\dots,t_m$ to $\phi_A^*(c)$ 
satisfies $c_j = \langle e^{t_j}u_j, (Ae^TA^*)^{-1}u_j\rangle $
for every $1\le j\le m.$
\end{lemma}

For completeness, we present Barthe's proof and to simplify notation we will continue
specializing our discussion to $c = \frac{n}{m} \vec{1}$. Consider maximizing the function $f$ over $\Rm$ defined
as:
$$f(t)=\frac nm\sum_{j=1}^m t_j - \log\det U$$
It is not hard to show that $f$ is concave (a short proof is given
in Lemma \ref{lem:convex}).
What is crucial is that if $t$ maximizes $f(t)$, then it must satisfy that 
the gradient of $f$ at $t$ vanishes. 
We can apply a well-known formula for the derivative of $\log\det$
(see e.g. \cite{Lax07}) and:
\[
0 = \frac{\partial f(t)}{\partial t_j}
= \frac nm - \tr\left( U^{-1}\frac{\partial U}{\partial t_j}\right)
\mper
\]
Also in our case:
\[
\frac{\rd Ae^TA^*}{\rd t_j}=\sum_{j=1}^m\frac{\rd e^{t_j}u_j\otimes u_j}{\rd
t_j}
= e^{t_j}u_j\otimes u_j\mper
\]
And so we conclude that the optimality condition is for all $j \in [m]$:
\[
0 = \frac nm - \tr\left( U^{-1} e^{t_j} u_j\otimes u_j\right)
= \frac nm - e^{t_j}\langle u_j,U^{-1}u_j\rangle\mper
\]
where the last step uses the identity $\tr(ABC) = \tr(BCA)$. Recall that $\langle u_j,U^{-1}u_j\rangle=\|Ru_j\|^2$
and so any optimal $t\in\Rm$
satisfies $e^{t_j} = \frac nm \|Ru_j\|^{-2}$ which is precisely the condition
we needed. 
Additionally, Barthe proves that the supremum of $f(t)$ is attained (this is a central step in the proof, which depends on the
particular $f$ given above). In order to give an algorithm for computing radial isotropic position, we will also need to
reason about $f$ and give effective bounds on the region where its supremum can be.

\section{Computing Radial Isotropic Position Efficiently}
\label{sec:sep}
\label{sec:efficient}
\label{sec:computerip}

Here we prove two important properties of the convex programming problem considered by Barthe,
that we will need in order to prove that the Ellipsoid method can solve it. We prove that if the optimum
is finite, there is a solution in a bounded region that is optimal. Also we establish a lower bound on how strictly convex the objective function is,
since we will need this to show that any candidate solution that is close enough to achieving the optimum value must
also be close to the optimum solution.

\subsection{Effective Bounds}

Here we prove bounds on the region in which an optimal solution can be found.
%
%
Our proof follows the same basic outline as in~\cite{BrascampLi76,Barthe98,CarlenLiLo04} but is self-contained.
We define $\phi_A\colon\Rm\to\mathbb{R}$ as
\[
\phi_A(t_1,\dots,t_m)= \log\det\left(\sum_{j=1}^m e^{t_j} u_j\otimes u_j\right)
\]
Given $c\in\Rm$, consider the optimization problem $\phi_A^*(c)
= \sup_{t\in\mathbb{R}^m}
\langle t,c\rangle - \phi_A(t_1,\dots,t_m).$
The function $\phi_A^*$ is the Legendre transform of $\phi_A.$ 
For convenience, we will write 
$\log\det (\sum_{j=1}^m e^{t_j} u_j\otimes u_j )
=\log\det(Ae^TA^*)$
where $T$ denotes the diagonal matrix with entries $t_1,\dots,t_m$ and 
$A^*$ is the transpose of~$A$ and $e^T$ denotes the matrix exponential of $T$ (i.e. a diagonal matrix in which
the $i^{th}$ entry on the diagonal is $e^{t_i}$). We also introduce the notation:

\begin{definition}
Let $d_I=\det (A_IA_I^*)$, $t_I=e^{\sum_{j\in I}t_j}$ and $D = \min_{I\colon d_I\ne 0} d_I$, where $A_I$ is the sub matrix whose columns are indexed by $I$.
\end{definition}

 When we use the subscript $I$ without further specification, we
will always mean a subset of $[m]$ of size $n.$ We will make repeated use of the Cauchy-Binet formula in this section:

\begin{fact}
Let $A, B^* \in \R^{m\times n}$. Then $\det(A B) = \sum_{I \colon |I| = n} \det(A_I) \det(B^*_I)$. 
\end{fact}

This generalizes the well-known identity that the determinant of the product of two matrices is the product of the determinants. We can apply this formula:

\begin{claim}
\label{cauchybinet}
$\det\Big(\sum_{j=1}^me^{t_j}u_j\otimes u_j\Big)
= \det(Ae^TA^*)=\sum_{I\subseteq[m],|I|=n}t_I d_I$
\end{claim}

We can now show that the mapping 
$\phi_A$ is convex, and hence $\phi_A^*(c)$ is concave (which we asserted in Section~\ref{sec:barthe}):
\begin{lemma}
\label{lem:convex}
The function $\phi_A$ is convex on $\mathbb{R}^m.$
\end{lemma}
\begin{proof}
Let $s,t\in\R^m.$ Then, applying the Chauchy-Schwarz inequality,
\begin{align*}
\phi\left(\frac{s+t}2\right)
& = \log\det(Ae^{(S+T)/2}A^*) = \log\left(\sum \sqrt{s_Id_I}\sqrt{t_Id_I}\right) \\
& \le \log\left(\sqrt{\sum s_I d_I}\sqrt{\sum t_I d_I}\right)  = \log\sqrt{\det(Ae^SA^*)} +\log\sqrt{\det(Ae^TA^*)} = \frac{\phi(s) +\phi(t)}2
\end{align*}
where the second equality uses Claim \ref{cauchybinet}. 
\end{proof}
Our main step is to show that if $c$ is contained in $K_A$ with ``sufficient
slack'', then the optimum is finite and we get a bound on the norm of an
optimizer. To state the condition we need, we will use a slightly
unconventional definition for how to dilate $K_A.$ This simplifies our
arguments (in part because it preserves the ``trivial'' constraints that the
coordinates sum to $n$ and are each in the interval $[0, 1]$):

\begin{definition}
Let $ C K_A$ denote the vectors $c$ whose coordinates sum to $n$ and are each
in the interval $[0, 1]$ and for all nonnegative directions $u$ with $u_{min}
= 0$, $C\max_{v \in K_A} \langle u, v  \rangle \geq \langle u, c \rangle $.
\end{definition}

\begin{lemma}
\label{yes}
Let $\alpha>0$ and suppose $c\in(1-\alpha)K_A.$ 
Then 
\begin{enumerate}
\item $\phi_A^*(c)<\log\frac1D$,
\item $t^*$ with 
$f(t^*)=\phi_A^*(c)$ satisfies
$\|t^*\|_\infty\le \frac 2\alpha
\log\frac1D
$
\end{enumerate}
\end{lemma}
\begin{proof}
From the assumption that $c\in K_A,$ it follows directly that 
$\sum_{i=1}^m c_i = n$ and that $c_i\in[0,1].$
Throughout this proof, let $f(t) = \langle c,t\rangle -\phi_A(t).$ Let $t\in\R^m.$ 
We need to upper bound $f(t).$ 
Note that we may assume that $\min_j t_j=0$ by adding a constant
$a\in\R$ to all coordinates without changing the value of $f(t)$. 
For notational convenience, assume that the coordinates of $t$ are sorted in
decreasing order $t_1\ge t_2\ge\dots\ge t_m=0.$ This is without loss of
generality since we can always apply a permutation to the columns of $A$ and
the coordinates of $t$ without changing the function value.

\begin{claim}\label{claim:lemma}
$f(t) 
\le \log\left(\frac 1D\right) - \alpha \max_{j=1}^m t_j$
\end{claim}
\begin{proof}
Let $I^*\subseteq[m]$ be the set of the $n$ pivotal vectors in
$[u_1,\dots,u_m],$ i.e., the indices of the vectors that are not in the span
of the vectors to the left of them.

By the monotonicity of the logarithm and
Claim \ref{cauchybinet}:
\[
\phi_A(t_1,\dots,t_m)
=\log\left(\sum t_I d_I \right)
\ge 
\log(t_{I^*}d_{I^*})
= \sum_{j\in I^*} t_j + \log(d_{I^*})
\ge \sum_{j\in I^*} t_j + \log\big(\min d_I\big)\mper
\]
Furthermore, we claim that
\begin{equation}\label{majorize}
\sum_{j\in I^*} t_j-\sum_{j=1}^m c_jt_j
\ge 
\alpha \max_j t_j\mper
\end{equation}
Together these two inequalities directly imply that
the statement of the claim. It therefore only remains to prove
(\ref{majorize}). First note that
$I^*$ maximizes $\langle\vec{1}_I,t\rangle=\sum_{j\in I}t_j$ 
among all $I$ such that $d_I\ne0.$ On the other hand, we know that $c\in
(1-\alpha)K_A.$ Hence $\langle c,t\rangle \le 
(1-\alpha) \langle\vec{1}_{I^*},t \rangle$
and this implies
\[
\sum_{j\in I^*} t_j-\sum_{j=1}^m c_jt_j 
\ge \alpha \sum_{j\in I^*} t_j
\ge \alpha t_1
\]
which establishes (\ref{majorize}).
\end{proof}
The previous claim
shows that as any $t_j$ tends to infinity, $f(t)$ tends to zero. Hence
$\phi_A^*(c)<\infty$ and, by the convexity of $\phi_A$, the supremum is
attained meaning that we can find $t^*$ such that $f(t^*)=\phi_A^*(c).$
But $f(t^*)=\phi_A^*(c)\ge f(0) 
=\log\det(AA^*)=\log\Big(\sum_I d_I\Big)\ge\log(\min_I d_I)$
where we used Claim \ref{cauchybinet} in the second inequality.
Combining this inequality with Claim~\ref{claim:lemma}, we conclude
\[
\max_j t_j^* \le \frac2\alpha
\log\Big(\frac1{\min_I d_I}\Big)
\mper
\]
\end{proof}

\subsection{Strict Convexity}
\sectionlabel{strict-convexity}
Here we prove that if a candidate solution $t$
is close to achieving the optimal value then it is also
close to the optimal solution $t^*$.
This is not a vacuous property since if a convex function $f$ is not strictly convex,
being close to the optimal value for the objective function does not imply that a
solution is close to the optimal solution.  

The catch is that our function $f$ is not
strictly convex on all of $\R^m.$ If $t_a$ denotes the vector obtained from $t$
by adding the constant $a$ to all coordinates in $t$, then for every $a,$
$f(t_a)=f(t)$ (where here we use the condition that $\sum_jc_j=n$). Hence,
there are points $t,t'$ at arbitrary distance that satisfy $f(t)=f(t').$
However, we can show that this is the only scenario in which the function
is not strictly convex.

\begin{definition}
Let us say that $s,t\in\Rm$ are $b$-separated if 
$\|(s+a{\bf 1})-t\|_\infty\ge b$ for every $a\in\R.$ 
Here, ${\bf 1}$ denotes
the all ones vectors.
\end{definition}

This definition leads to the next lemma. 
Using this lemma we will later argue that whenever $f(t)$ is very close
to optimal, then $t$ itself cannot be separated from an optimal solution by
much.
\begin{lemma}
\label{lem:strict}
Let $s,t\in\Rm$ be any two $b$-separated points for some $b>0.$ 
Assume all coordinates of $s,t$ are non-negative and that for every $i,j\in[m]$ there exists $S\subseteq[m]$ with
$|S|=n-1$ such that $d_{S\cup\Set{i}} \ne 0$ and $d_{S\cup\Set{j}\ne0}.$ Then,
\begin{equation*}
\phi_A\left(\frac{s+t}2\right)\le\frac{\phi(s)+\phi(t)}2
-b^2\cdot \exp\left(-(n+1)(\|s\|_\infty+\|t\|_\infty)\right)
\frac{\min_I d_I^2}{\det(AA^*)}\mper
\end{equation*}
\end{lemma}

\begin{proof}
As we did in the proof of Lemma~\ref{lem:convex}, we will apply the
Cauchy-Schwarz inequality to the vectors $u,v$ indexed by 
$I\subseteq[m],|I|=n$ and defined as
\[
u_I=\sqrt{e^{\sum_{j\in I} s_j} d_I}\mcom\qquad
v_I=\sqrt{e^{\sum_{j\in I} t_j} d_I}\mper
\]
We'd like to determine how much slack we have in this inequality.
Let us therefore lower bound
\[
1-\left(\frac{\langle u,v\rangle}{\|u\|\|v\|}\right)^2
= \frac{\|u\|^2\|v\|^2-\langle u,v\rangle^2}{\|u\|^2\|v\|^2}
= \frac{\frac12\sum_{I\ne J}(u_Iv_J-u_Jv_I)^2}{\|u\|^2\|v\|^2}\mcom
\]
where the last step is Lagrange's identity.
Now write $t_j=s_j+2a_j.$ By the assumption that $s,t$ are $b$-separated we
must have that $\max_{i,j\in[m],i\ne j}\left|a_i-a_j\right|\ge b.$ 
Let $i,j$ be a pair of indices achieving 
the maximum. Without loss of generality assume that $a_j\ge a_i+b.$
Let $S\subseteq[m]$ be a set of size $|S|=n-1$ such that 
$I=S\cup\{i\}$ and $J=K\cup\{j\}$ satisfy $d_I\ne 0$ and $d_J\ne 0.$ 
Such a set~$S$ must exist by our assumption.
Then:
\begin{align*}
(u_Iv_J - u_Jv_I)^2
&=
\left((e^{\frac{s_i+t_{j}}2}-e^{\frac{s_{j}+t_i}2})
e^{\sum_{j\in S}\frac{s_j+t_j}2}\right)^2 d_Id_J \\
&= \left( (e^{a_{j}}-e^{a_i})e^{\frac{s_i+s_{j}}2}e^{\sum_{j\in S}\frac{s_j+t_j}2} \right)^2  d_Id_J  \ge 
 (e^{a_{j}}-e^{a_i})^2 d_Id_J \\
\end{align*}

\vspace{-1.5pc}

\noindent where the inequality follows because $s_i,t_i\ge0$ for all $i$. On the other hand,
$(e^{a_{j}}-e^{a_i})^2 
= (e^b-1)^2e^{2a_i}.$
But $e^x-1\ge x$ and $a_j\ge-\|s\|_\infty.$ Thus:
\[
(u_Iv_J-u_Jv_I)^2\ge\gamma 
\text{ with } \gamma=b^2 \min_I d_I^2\cdot e^{-\|s\|_\infty}\mper
\] 
Therefore:
\begin{equation}\label{eq:slack}
\left(\frac{\langle u,v\rangle}{\|u\|\|v\|}\right)^2
\le 1-\frac\gamma{\|u\|\|v\|}
\mper
\end{equation}
On the other hand
\[
\|u\|=\sqrt{\det(Ae^SA^*)}\le 
e^{n\|s\|_\infty}\det(AA^*)\mcom\qquad
\|v\|=\sqrt{\det(Ae^TA^*)}\le
e^{n\|t\|_\infty}\det(AA^*)\mper
\] 
Taking logarithms on both sides of~(\ref{eq:slack}), we get
\[
\log\left(\frac{\langle u,v\rangle}{\|u\|\|v\|}\right)
\le -\frac12\gamma \frac{e^{-n(\|s\|_\infty+\|t\|_\infty)}}{\det(AA^*)}\mcom
\]
where we used that $\log(1-x)\le-x$ for all $1>x>0.$
\end{proof}

\subsection{An Algorithm}

Our next theorem gives a polynomial time algorithm for computing the radial
isotropic position. The assumptions are slightly stronger than simply asking
that $c\in K_A.$

\begin{theorem}\label{thm:efficient}
Let $\epsilon>0$ and $\alpha>0.$
Let $A=[u_1,\dots,u_m]\in\R^{m\times n}$ with $m\ge n$ and
$\mathrm{rank}(A)=n.$
Further assume that for every $i,j\in[m]$ there exists $S\subseteq[m]$ with
$|S|=n-1$ such that
$d_{S\cup\Set{i}} \ne 0$ and $d_{S\cup\Set{j}\ne0}.$ 
Then, given $A$ and
any point $c\in (1-\alpha)K_A,$ we can compute a
 $n\times n$ matrix $R$ such that
\begin{equation*}\label{eq:approx}
\sum_{j=1}^m
c_j \left(\frac{Ru_j}{\|Ru_j\|}\right)\otimes 
\left(\frac{Ru_j}{\|Ru_j\|}\right)
= \mathrm{Id}_{\Rn} + J 
\mcom
\end{equation*}
where $\|J\|_\infty\le\epsilon.$
The running time of our algorithm is 
polynomial in $1/\gamma,\log(1/\epsilon)$ and $L$ where $L$ is an upper bound
on the bit complexity of the input $A$ and $c.$
\end{theorem}

\begin{proof}
We will apply the Ellipsoid method as described in~\cite{Nemirovski05} 
(Theorem 4.1.2.) to
solve the optimization problem $\sup_{t\in \Rm} \langle c,t\rangle-\phi_A(t)$
 over the set of all $t\in\Rm$ satisfying $\|t\|_\infty\le B$ where $B$ is the
parameter from Lemma~\ref{yes} and $t_j\ge0$ for all $j\in[m].$
Let $s$ denote the solution computed by
the Ellipsoid method and suppose we have $|f(s)-f(t^*)|\le \delta^2$ 
where 
\[
\delta \le \frac{\epsilon'\cdot \min_I d_I}{e^{O(Bn)}\det(AA^*)}\mcom
\]
and $\epsilon'$ is a sufficiently small quantity that we will bound later.
With $\delta$ chosen this small it follows from Lemma~\ref{lem:strict} 
that $t^*$ and $s$ cannot be $\delta$-separated. (Otherwise 
$\frac{s+t}2$ would give a solution improving the optimum.) Here, we used the
fact that $\sum_{j\in I} s_j\le Bn$ for every $I\subseteq[m],|I|=n$ 
and therefore
\[
e^{\phi(s)}\le e^{\log(e^{Bn}\det(AA^*))}
= e^{Bn}\det(AA^*),
\]
Similarly, we get the same bound for for $\phi(t^*).$
Hence, we conclude that $s$ must be $\delta$-close to an optimal 
solution in each coordinate. 
This implies (using standard perturbation bounds for the inverse
of a matrix) that  the optimality conditions from Lemma~\ref{lem:opt} are
approximately satisfied for~$s$  in the sense that
\[
e^{s_j} 
=\frac{(1+\eps_j)c_j}{\langle u_j,(Ae^SA^*)^{-1}u_j\rangle}\mcom
\]
with $\eps_j\in[-\epsilon',\epsilon'].$
Consider the positive definite matrix
$M=\sum_{j} e^{s_j} u_j\otimes u_j.$ Its inverse square root $R=M^{-1/2}$ 
satisfies
\[
\Id = 
\sum_{j=1}^m c_j \frac{Ru_j\otimes Ru_j}{\|Ru_j\|^2}
+ \sum_{j=1}^m \eps_jc_j \frac{Ru_j\otimes Ru_j}{\|Ru_j\|^2}\mper
\]
Let $J= \sum_{j=1}^m \epsilon_jc_j \frac{Ru_j\otimes Ru_j}{\|Ru_j\|^2}$ denote
the error term above. 
It is not hard to show that for
$\epsilon'=\epsilon/\exp(\poly(L))),$ we have that $\|J'\|_\infty\le\epsilon.$
Since the dependence on $1/\delta$ in the Ellipsoid method is logarithmic, the 
running time remains polynomial in $L,1/\alpha$ and $\log(1/\epsilon).$
\end{proof}

\end{document}